\documentclass[12pt,letterpaper]{article}
\usepackage{amsmath,amsfonts,amsthm,amssymb,tabls,stmaryrd,}
\usepackage{graphicx,relsize,cite}

\theoremstyle{plain}

\numberwithin{equation}{section}
\newtheorem{thm}{Theorem}[section]

\newenvironment{exam}[1]%
{\begin{flushleft}\textbf{Example #1}.\enspace}%
{\end{flushleft}}

\newcommand{\complex}{{\mathbb C}}

\newcommand{\ascript}{{\mathcal A}}
\newcommand{\dscript}{{\mathcal D}}
\newcommand{\escript}{{\mathcal E}}
\newcommand{\pscript}{{\mathcal P}}
\newcommand{\rscript}{{\mathcal R}}
\newcommand{\tscript}{{\mathcal T}}

\newcommand{\itre}{\mathop{\mathit{Re}}}

\newcommand{\xhat}{\widehat{x}}
\newcommand{\pscripthat}{\widehat{\pscript}}
\newcommand{\Omegahat}{\widehat{\Omega}}
\newcommand{\escripthat}{\widehat{\escript}}
\newcommand{\abar}{\overline{a}}
\newcommand{\cbar}{\overline{c}}
\newcommand{\atilde}{\widetilde{a}}

\newcommand{\ab}[1]{\left|#1\right|}

\newcommand{\brac}[1]{\left\{#1\right\}}
\newcommand{\paren}[1]{\left(#1\right)}
\newcommand{\sqbrac}[1]{\left[#1\right]}
\newcommand{\elbows}[1]{{\left\langle#1\right\rangle}}
\newcommand{\floors}[1]{{\left\lfloor#1\right\rfloor}}
\newcommand{\ket}[1]{{\left|#1\right>}}
\newcommand{\bra}[1]{{\left<#1\right|}}

\begin{document}

\title{A COVARIANT CAUSAL SET\\APPROACH TO\\DISCRETE QUANTUM GRAVITY
}
\author{S. Gudder\\ Department of Mathematics\\
University of Denver\\ Denver, Colorado 80208, U.S.A.\\
sgudder@du.edu
}
\date{}
\maketitle

\begin{abstract}
A covariant causal set (c-causet) is a causal set that is invariant under labeling. Such causets are well-behaved and have a rigid geometry that is determined by a sequence of positive integers called the shell sequence. We first consider the microscopic picture. In this picture, the vertices of a
c-causet have integer labels that are unique up to a label isomorphism. This labeling enables us to define a natural metric $d(a,b)$ between time-like separated vertices $a$ and $b$. The time metric $d(a,b)$ results in a natural definition of a geodesic from $a$ to $b$. It turns out that there can be $n\ge 1$ such geodesics. Letting $a$ be the origin (the big bang), we define the curvature $K(b)$ of $b$ to be $n-1$. Assuming that particles tend to move along geodesics, $K(b)$ gives the tendency that vertex $b$ is occupied. In this way, the mass distribution is determined by the geometry of the c-causet. We next consider the macroscopic picture which describes the growth process of c-causets. We propose that this process is governed by a quantum dynamics given by complex amplitudes. At present, these amplitudes are unknown. But if they can be found, they will determine the (approximate) geometry of the c-causet describing our particular universe. As an illustration, we present a simple example of an amplitude process that may have physical relevance. We also give a discrete analogue of Einstein's field equations.
\end{abstract}

\section{Introduction}  % Section 1
The causal set (causet) approach to discrete quantum gravity has been studied by various investigators \cite{gtw09,gud13,hen09,sor03,sur11,wal13}. Unlike previous sequential growth models, in our approach the basic elements are a special type of causet called a covariant causet (c-causet). A c-causet is defined to be a causet that is invariant under labeling. That is, two different labelings of a
 c-causet are label isomorphic. This is a strong restriction which says that the elements of the causet have a unique ``birth order'' up to isomorphism. The restriction to c-causets provides great mathematical simplifications. For example, every c-causet possess a unique history and has precisely two covariant offspring. It follows that there are $2^{n-1}$ c-causets of cardinality $n$. There are also physical reasons to consider c-causets. An arbitrary causet can have a very wild geometry that has no relationship to our universe which is isotropic and homogeneous in the large. On the other hand, a c-causet has a rather mild and well-behaved geometry. This rigid geometry is determined by a sequence of positive integers called the shell sequence.
 
 After our discussion of c-causets, in the following section, we consider the microscopic picture. This picture considers the detailed structure of a
 c-causet in which the vertices have integer labels that are unique up to a label isomorphism. This labeling enables us to define a natural metric distance $d(a,b)$ between time-like separated vertices $a$ and $b$. The time metric $d(a,b)$ results in a natural definition of a geodesic from $a$ to $b$. It turns out that there can be $n\ge 1$ such geodesics. Letting $a$ be the origin (the big bang), we define the curvature $K(b)$ at $b$ to be $n-1$. Assuming that particles tend to move along geodesics, $K(b)$ gives the tendency that vertex $b$ is occupied. In this way, the mass distribution is determined by the geometry of the c-causet. This is counter to the usual view in general relativity where the curvature is determined by the mass distribution. We thus reach the conclusion that the properties of a universe are dictated by the geometry, that is, the shell sequence of the underlying c-causet. Various small examples of shell sequence toy universes are considered. For instance, an exponential growth (inflation) appears to result in a fairly uniform small curvature indicating a homogeneous mass distribution, while a sudden contraction results in large curvatures reminiscent of a black hole. Other examples indicated clusters of masses like galaxies. Of course, these examples are only toy universes and computer simulations would be required for more definitive models.
 
 This now leads to the big question. How do we determine the geometry of our particular universe? We do not expect to find an exact answer to this question. We only hope to find probabilities for various competing geometries and this is where quantum-like theory comes into play. To study this big question, we consider the macroscopic picture which describes the growth process of c-causets. We propose that this process is governed by a quantum dynamics given by complex amplitudes. At present, these amplitudes are unknown. But if they can be found, they will determine the probabilities for geometries of c-causets. As an illustration we shall present a simple example of an amplitude process that may have physical relevance.
 
In the macroscopic picture, we find it appropriate not to consider c-causets individually, but in pairs called twins which are two different c-causets with the some history. The set of twin c-causets has the structure of a discrete 4-manifold which brings us closer to traditional general relativity theory. Finally, we present a discrete analogue of Einstein's field equations. To summarize, this work suggests, contrary to conventional wisdom, that the microscopic picture is described by geometry (discrete general relativity) while the macroscopic picture is described by a discrete quantum amplitude process.

\section{Covariant Causets} % Section 2
In this article we call a finite partially ordered set a \textit{causet}. If $a,b$ are elements of a causet $x$, we interpret the order $a<b$ as meaning that $b$ is in the causal future of $a$ and $a$ is in the causal past of $b$. An element $a\in x$ is \textit{maximal} if there is no $b\in x$ with $a<b$. If
$a,b\in x$, we say that $a$ and $b$ are \textit{comparable} if $a<b$ or $b<a$. If $a<b$ and there is no $c\in x$ with $a<c<b$, then $a$ is a
\textit{parent} of $b$ and $b$ is a \textit{child} of $a$. A \textit{chain} is a set of mutually comparable elements of $x$. The \textit{height} $h(a)$ of
$a\in x$ is the cardinality of the longest chain in $x$ whose largest element is $a$. We denote the cardinality of $x$ by $\ab{x}$.

If $x$ and $y$ are causets with $\ab{y}=\ab{x}+1$, then $x$ \textit{produces} $y$ if $y$ is obtained from $x$ by adjoining a single maximal element $a$ to $x$. In this case we write $y=x\uparrow a$ and use the notation $x\to y$. If $x\to y$, we also say that $x$ is a \textit{producer} of $y$ and
$y$ is an \textit{offspring} of $x$. In general, $x$ may produce many offspring and $y$ may be the offspring of many producers. (It is suggested that throughout this article, the reader should draw diagrams to illustrate and understand the various concepts.)

A \textit{labeling} for a causet $x$ is a bijection $\ell\colon x\to\brac{1,2,\ldots ,\ab{x}}$ such that $a,b\in x$ with $a<b$ implies that $\ell (a)<\ell (b)$. A \textit{labeled cause} is a pair $(x,\ell )$ where $\ell$ is a labeling of $x$. For simplicity, we frequently write $x=(x,\ell )$ and call $x$ an
$\ell$-\textit{causet}. Two $\ell$-causets $x$ and $y$ are \textit{isomorphic} if there exists a bijection $\phi\colon x\to y$ such that $a<b$ if and only if $\phi (a)<\phi (b)$ and $\ell\sqbrac{\phi (a)}=\ell (a)$ for all $a\in x$. Isomorphic $\ell$-causets are considered identical as $\ell$-causets. It is not hard to show that any causet can be labeled to form an $\ell$-causet. In general, a causet can be labeled in many nonisomorphic ways but there are exceptions and these are the ones of importance in this work. A causet is \textit{covariant} if it has a unique labeling (up to $\ell$-causet isomorphism). This is a strong restriction which says that the elements of the causet have a unique ``birth order'' up to isomorphisms. We call a covariant causet a c-\textit{causet}. We denote the set of c-causets with cardinality $n$ by $\pscript _n$ and the set of all c-causets by $\pscript$. It is easy to show that any $x\in\pscript$ with $\ab{x}>1$ has a unique producer. Moreover, it is shown in \cite{gudp13} that any c-causet has precisely two covariant offspring. It follows that $\ab{\pscript _n}=2^{n-1}$, $n=1,2,\ldots\,$. The following result is proved in \cite{gudp13}.

\begin{thm}       % Theorem 2.1
\label{thm21}
A causet $x$ is covariant if and only if $a,b\in x$ are comparable whenever $h(a)\ne h(b)$.
\end{thm}

For $x\in\pscript$, let
\begin{equation*}
S_j(x)=\brac{a\in x\colon h(a)=j}, j=0,1,2,\ldots
\end{equation*}
We call the sets $S_j(x)\subseteq x$ \textit{shells} and the sequence of integers $s_j(x)=\ab{S_j(x)}$, $j=0,1,\ldots$ is the \textit{shell sequence} for $x$. Of course, elements in the same shell are incomparable and we say they are \textit{space-like separated}. It follows from Theorem~\ref{thm21} that two elements in different shells are comparable and we say they are \textit{time-like} separated. We conclude that a c-causet is uniquely determined by its shell sequence. We think of $\brac{s_j(x)}$ as describing the ``shape'' or geometry of $x$. Mathematically a shell sequence $\brac{s_j(x)}$ is a sequence of positive integers satisfying $\sum s_j(x)=\ab{x}$, where the order of the terms in the sequence is taken into account. To illustrate this, we list the shell sequences for the c-causets with cardinality $n=1,2,3,4,5$.
\begin{align*}
n=1\colon&(1)\\
n=2\colon&(1,1),(2)\\
n=3\colon&(1,1,1),(1,2),(2,1),(3)\\
n=4\colon&(1,1,1,1),(1,1,2),(1,2,1),(2,1,1),(2,2),(1,3),(3,1),(4)\\
n=5\colon&(1,1,1,1,1),(1,1,1,2),(1,1,2,1),(1,2,1,1),(2,1,1,1),(1,1,3),\\
  &(1,3,1),(3,1,1),(1,4),(4,1),(2,3),(3,2),(1,2,2),(2,1,2),\!(2,2,1),\!(5)
  \end{align*}
Each of these shell sequences represents a unique c-causet and as we have previously mentioned the number of sequences with sum $n$ is $2^{n-1}$.

\section{Microscopic Picture} % Section 3
In the microscopic picture we examine the detailed structure of a c-causet $x$. We view $x$ as a framework or scaffolding for a possible universe. The vertices of $x$ represent small cells that can be empty or occupied by a particle. The shell sequence that determines $x$ gives the geometry of the framework. In order to describe this universe, we would like to find out how particles move and which vertices they are likely to occupy. We accomplish this by introducing a distance or metric on $x$.

Let $x=\brac{a_1,a_2,\ldots ,a_n}\in\pscript _n$ where the subscript $i$ of $a_i$ is the label of the vertex. For $a_{i_1}, a_{i_m}\in x$ with
$a_{i_1}<a_{i_m}$, a \textit{path} from $a_{i_1}$ to $a_{i_m}$ is a sequence $a_{i_1}<a_{i_2}<\cdots <a_{i_m}$ where $a_{i_j}$ is a parent of
$a_{i_{j+1}}$, $j=1,\ldots ,m-1$. We can think of a path from $a_{i_1}$ to $a_{i_m}$ as a sequence in $x$ starting with $a_{i_1}$ and moving along successive shells until $a_{i_m}$ is reached. If
\begin{equation*}
\gamma =a_{i_1}a_{i_2}\dots a_{i_m}
\end{equation*}
is a path, we define the \textit{length} of $\gamma$ by
\begin{equation*}
\ell (\gamma )=\sqbrac{\sum _{j=2}^m(i_j-i_{j-1})^2}^{1/2}
\end{equation*}
Of course, there are a variety of definitions that one can give for the length of a path, but this is one of the simplest nontrivial choices. For $a,b\in x$ with $a<b$, a \textit{geodesic} from $a$ to $b$ is a path from $a$ to $b$ that has the shortest length. Clearly, if $a<b$, then there is at least one geodesic from $a$ to $b$. If $a,b\in x$ are time-like separated (comparable) and $a<b$ say, then the \textit{distance} $d(a,b)$ is the length of a geodesic from $a$ to $b$. We have that $d(a,b)=d(b,a)$ by definition and we do not define $d(a,b)$ if $a$ and $b$ are incomparable. (We could define $d(a,a)=0$ but this is not needed.) The next result shows that the triangle inequality holds when applicable so $d(a,b)$ has the most important property of a metric. Since $d(a,b)$ is based on the time between ``births'' of $a$ and $b$, we also call $d(a,b)$ the \textit{time metric}.

\begin{thm}       % Theorem 3.1
\label{thm31}
If $a<c<b$, then $d(a,b)\le d(a,c)+d(c,b)$.
\end{thm}
\begin{proof}
Let $a_{i_1}a_{i_2}\cdots a_{i_m}$ be a geodesic from $a$ to $b$ $(a_{i_1}=a, a_{i_m}=b)$. Also, let $a_{j_1}a_{j_2}\cdots a_{j_r}$ be a geodesic from $a$ to $c$ and $a_{k_1}a_{k_2}\cdots a_{k_s}$ be a geodesic from $c$ to $b$. As before the subscripts are the labels of the corresponding vertices. We have by Minkowski's inequality that
\begin{align*}
d(a,b)&=\sqbrac{\sum (i_n-i_{n-1})^2}^{1/2}\le\sqbrac{\sum (j_n-j_{n-1})^2+\sum (k_n-k_{n-1})^2}^{1/2}\\
  &\le\sqbrac{\sum (j_n-j_{n-1})^2}^{1/2}+\sqbrac{\sum (k_n-k_{n-1})^2}^{1/2}=d(a,c)+d(c,b)\ \qedhere
\end{align*}
\end{proof}

A \textit{subpath} of a path $a_i,a_{i_2}\cdots a_{i_m}$ is a subset of $\brac{a_{i_1}, a_{i_2},\ldots ,a_{i_m}}$ that is again a path. The next result shows that once we have a geodesic we can take subpaths to form other geodesics.

\begin{thm}       % Theorem 3.2
\label{thm32}
A subpath of a geodesic is a geodesic.
\end{thm}
\begin{proof}
Let $a_{i_1}a_{i_2}\cdots a_{i_n}$ be a geodesic. If $a_{j_1}a_{j_2}\cdots a_{j_m}$ is a subpath that is not a geodesic then there is a smaller length path $a_{j_1}a_{k_2}\cdots a_{k_r}a_{j_m}$ from $a_{j_1}$ to $a_{j_m}$. But then
\begin{equation*}
a_{i_1}\ldots a_{j_1}a_{k_1}\cdots a_{k_r}a_{j_m}\cdots a_{i_n}
\end{equation*}
is a path from $a_{i_1}$ to $a_{i_n}$ with smaller length then $a_{i_1}a_{i_2}\cdots a_{i_n}$ which is a contradiction
\end{proof}

We have seen that the shell sequence determines the metric $d$. It is interesting that the converse holds in that the metric determines the shell sequence and hence the geometry of $x\in\pscript$. In fact, all we need to know is when $d(a,b)$ is defined. Let $x=\brac{a_1,a_2,\ldots, a_n}$ where the subscripts are vertex labels. Let $j_0$ be the smallest integer for which $d(a_1,a_{j_0})$ is undefined. If $j_0=n$, then $s_0(x)=n$ and we are finished. Otherwise, $s_0(x)=j_0<n$. Now $a_{j_0+1}\in S_1(x)$ and let $j_1>j_0+1$ be the smallest integer for which
$d(a_{j_0+1},a_{j_1})$ is undefined. If $j_1=n$, then $s_1(x)=n-j_0$ and we are finished. Otherwise, $s_1(x)=j_1-j_0$. Continuing by induction, we obtain the entire shell sequence $(s_0(x),\ldots, s_m(x))$.

Suppose we have a c-causet $x$ with shell sequence $(s_0(x),s_1(x),\ldots ,s_m(x))$. We view this as a sequential growth process in which vertices are born one at a time. We start with the vertex labeled 1 and new vertices are sequentially born until the first shell is filled: $1,2,\ldots ,s_0(x)$. Then the second shell is filled: $s_0(x)+1,\ldots ,s_0(x)+s_1(x)$. The process continues until the last shell is filled. We can thus represent $x$ by the sequence of integers:
\begin{align*}
x=&(1,2,\ldots ,s_0(x); s_0(x)+1,\ldots ,s_0(x)+s_1(x);\cdots ;\\
  &s_0(x)+\cdots +s_{m-1}(x)+1,\ldots ,s_0(x)+\cdots +s_m(x))
\end{align*}
In the sequel we shall employ the notation of the previous sentence to describe an $x\in\pscript$.

As we shall see, there may be more than one geodesic from $a$ to $b$ when $a<b$. If there are $j$ geodesics from vertex~1 to vertex $n$, we define the \textit{curvature} $K(n)$ at $n$ to be $K(n)=j-1$. One might argue that the curvature should be a local property and should not depend so heavily on vertex~1 which could be a considerable distance away. However, if there are a lot of geodesics from 1 to $n$, then by
Theorem~\ref{thm32}, there are also a lot of geodesics from other vertices to $n$. Thus, the definition of curvature is not so dependent on the initial vertex~1 as it first appears. Assuming that particles tend to move along geodesics, we see that $K(n)$ gives a measure of the tendency for vertex $n$ to be occupied.

As an example, consider the simple toy universe with shell sequence $(1,2,3,4,5,4,3,2,1)$. The labels of the vertices become
\begin{equation*}
x\!=\!(1;2,3;4,5,6;7,8,9,10;11,\!12,13,\!14,\!15;16,17,18,19;20,21,22;23,24;25)
\end{equation*}
The c-causet $x$ represents a toy universe that expands uniformly and then contracts uniformly. Of course, the paths 1--2 and 1--3
are geodesics so $d(1,2)=1$, $d(1,3)=2$. There are two paths from 1 to 4 given by 1--2--4 and 1--3--4. We conclude that both of these paths are geodesics and $d(1,4)=\sqrt{5}$. There are two paths for 1 to 5 given by 1--2--5 and 1--3--5. The first path has length $\sqrt{10}$ and the second has length $\sqrt{8}$. Hence, 1--3--5 is a geodesic and $d(1,5)=\sqrt{8}$. Similarly, there are two paths from 1 to 6 given by 1--2--6 and 1--3--6. The path 1--3--6 is the only geodesic and $d(1,6)=\sqrt{13}$. For a last example, consider the paths from 1 to 11. There are 24 paths from 1 to 11, but only the paths 1--3--6--8--11, 1--3--5--8--11, 1--3--6--9--11 are geodesics so $d(1,11)=\sqrt{26}$. The following two tables summarize the distances and curvatures for the vertices of the c-causet $x$.
\vskip 2pc

% Table 1
{\hskip - 4pc
\begin{tabular}{c|c|c|c|c|c|c|c|c|c|c|c|c|c}
$i$&$2$&$3$&$4$&$5$&$6$&$7$&$8$&$9$&$10$&$11$&$12$&$13$&$14$\\
\hline
$d(1,i)$&$1$&$2$&$\sqrt{5}$&$\sqrt{8}$&$\sqrt{13}$&$\sqrt{12}$&$\sqrt{17}$
  &$\sqrt{22}$&$\sqrt{29}$&$\sqrt{26}$&$\sqrt{31}$&$\sqrt{38}$&$\sqrt{45}$\\
\end{tabular}}
\vskip 2pc
{\hskip - 4pc
\begin{tabular}{c|c|c|c|c|c|c|c|c|c|c|c}
$i$&$15$&$16$&$17$&$18$&$19$&$20$&$21$&$22$&$23$&$24$&$25$\\
\hline
$d(1,i)$&$\sqrt{54}$&$\sqrt{47}$&$\sqrt{54}$&$\sqrt{61}$&$\sqrt{70}$&$\sqrt{63}$&$\sqrt{70}$
  &$\sqrt{77}$&$\sqrt{72}$&$\sqrt{79}$&$\sqrt{76}$\\
\noalign{\bigskip}
\multicolumn{12}{c}%
{\textbf{Table 1}}\\
\end{tabular}}
\vskip 2pc

% Table 2
{\begin{tabular}{c|c|c|c|c|c|c|c|c|c|c|c|c|c|c}
$i$&$1$&$2$&$3$&$4$&$5$&$6$&$7$&$8$&$9$&$10$&$11$&$12$&$13$&$14$\\
\hline
$K(i)$&$-1$&$0$&$0$&$1$&$0$&$0$&$0$&$1$&$0$&$0$&$2$&$0$&$1$&$0$\\
\end{tabular}}
\vskip 2pc
{\begin{tabular}{c|c|c|c|c|c|c|c|c|c|c|c}
$i$&$15$&$16$&$17$&$18$&$19$&$20$&$21$&$22$&$23$&$24$&$25$\\
\hline
$K(i)$&$0$&$2$&$2$&$0$&$1$&$5$&$3$&$1$&$5$&$9$&$5$\\
\noalign{\bigskip}
\multicolumn{12}{c}%
{\textbf{Table 2}}\\
\end{tabular}}
\vskip 2pc

In this toy universe, vertices 4, 8 and 13 might represent toy planets about a toy star 11. Vertices 16 and 17 might represent a toy double star and 20, 21 might be large toy stars with toy planets 19, 21. Finally, vertices 23,24, 25 might represent a toy black hole. This example indicates that an expansion followed by a contraction results in some relatively large curvatures. 

Theorem~\ref{thm32} can be used to construct an algorithm for determining geodesics recursively. A geodesic from $a$ to $b$ can be formed by taking an appropriate vertex $c$ in the shell immediately prior to the shell containing $b$ and constructing a geodesic from $a$ to $c$ together with the edge $cb$. This algorithm gives a necessary but not a sufficient condition for a geodesic. That is, a geodesic followed by another connecting geodesic need not form a geodesic. For instance, in the previous example, the paths 1--3--6 and 6--7 are geodesics, but 1--3--6--7 is not a geodesic.

We now give an example which indicates that exponential growth (inflation) results in a fairly uniform low curvature. Let $y$ be the c-causet with shell sequence $(1,2,4,8,16)$. The resulting vertex labeling becomes:
\begin{equation*}
y=(1;2,3;4,5,6,7;8,9,10,11,12,13,14,15;16,17,\ldots ,31)
\end{equation*}
The following two tables summarize the distances and curvatures for the vertices of $y$.
\vskip 2pc

% Table 3
{\hskip -2pc
\begin{tabular}{c|c|c|c|c|c|c|c|c|c|c|c}
$i$&$2$&$3$&$4$&$5$&$6$&$7$&$8$&$9$&$10$&$11$&$12$\\
\hline
$d(1,i)$&$1$&$2$&$\sqrt{5}$&$\sqrt{8}$&$\sqrt{13}$&$\sqrt{20}$&$\sqrt{17}$
  &$\sqrt{22}$&$\sqrt{29}$&$6$&$\sqrt{45}$\\
\end{tabular}}
\vskip 2pc
{\hskip -2pc
\begin{tabular}{c|c|c|c|c|c|c|c|c|c|c}
$i$&$13$&$14$&$15$&$16$&$17$&$18$&$19$&$20$&$21$&$22$\\
\hline
$d(1,i)$&$\sqrt{56}$&$\sqrt{69}$&$\sqrt{84}$&$\sqrt{61}$&$\sqrt{70}$&$9$&$\sqrt{92}$&$\sqrt{105}$&$\sqrt{118}$&$\sqrt{133}$\\
\end{tabular}}
\vskip 2pc
{\hskip -2pc
\begin{tabular}{c|c|c|c|c|c|c|c|c|c}
$i$&$23$&$24$&$25$&$26$&$27$&$28$&$29$&$30$&$31$\\
\hline
$d(1,i)$&$\sqrt{148}$&$\sqrt{164}$&$\sqrt{184}$&$\sqrt{205}$&$\sqrt{228}$&$\sqrt{253}$&$\sqrt{280}$&$\sqrt{309}$&$\sqrt{340}$\\
\noalign{\bigskip}
\multicolumn{10}{c}%
{\textbf{Table 3}}\\
\end{tabular}}
\vskip 2pc

% Table 4
{\hskip -2pc
\begin{tabular}{c|c|c|c|c|c|c|c|c|c|c|c|c|c|c|c|c|c}
$i$&$1$&$2$&$3$&$4$&$5$&$6$&$7$&$8$&$9$&$10$&$11$&$12$&$13$&$14$&$15$&$16$&$17$\\
\hline
$K(i)$&$-1$&$0$&$0$&$1$&$0$&$0$&$0$&$1$&$0$&$1$&$0$&$0$&$0$&$0$&$0$&$1$&$1$\\
\end{tabular}}
\vskip 2pc
{\begin{tabular}{c|c|c|c|c|c|c|c|c|c|c|c|c|c|c}
$i$&$18$&$19$&$20$&$21$&$22$&$23$&$24$&$25$&$26$&$27$&$28$&$29$&$30$&$31$\\
\hline
$K(i)$&$1$&$0$&$1$&$0$&$1$&$0$&$0$&$0$&$0$&$0$&$0$&$0$&$0$&$0$\\
\noalign{\bigskip}
\multicolumn{14}{c}%
{\textbf{Table 4}}\\
\end{tabular}}
\vskip 2pc

Our last example indicates that periodic shell sequence behavior results in periodic curvature behavior. Let $z$ be the c-causet with shell sequence $(1,2,3,2,1,2,3,2,1)$ and vertex labeling
\begin{equation*}
z=(1;2,3;4,5,6;7,8;9;10,11;12,13,14;15,16;17)
\end{equation*}
The next tables summarize the distance and curvatures for the vertices of $z$.
\vskip 2pc

% Table 5
{\begin{tabular}{c|c|c|c|c|c|c|c|c}
$i$&$2$&$3$&$4$&$5$&$6$&$7$&$8$&$9$\\
\hline
$d(1,i)$&$1$&$2$&$\sqrt{5}$&$\sqrt{8}$&$\sqrt{13}$&$\sqrt{12}$&$\sqrt{17}$
  &$\sqrt{22}$\\
\end{tabular}}
\vskip 2pc
{\begin{tabular}{c|c|c|c|c|c|c|c|c}
$i$&$10$&$11$&$12$&$13$&$14$&$15$&$16$&$17$\\
\hline
$d(1,i)$&$\sqrt{23}$&$\sqrt{26}$&$\sqrt{27}$&$\sqrt{30}$&$\sqrt{35}$&$\sqrt{34}$&$\sqrt{39}$&$\sqrt{38}$\\
\noalign{\bigskip}
\multicolumn{8}{c}%
{\textbf{Table 5}}\\
\end{tabular}}
\vskip 2pc

% Table 6
{\hskip -2pc
\begin{tabular}{c|c|c|c|c|c|c|c|c|c|c|c|c|c|c|c|c|c}
$i$&$1$&$2$&$3$&$4$&$5$&$6$&$7$&$8$&$9$&$10$&$11$&$12$&$13$&$14$&$15$&$16$&$17$\\
\hline
$K(i)$&$-1$&$0$&$0$&$1$&$0$&$0$&$0$&$1$&$0$&$0$&$0$&$1$&$0$&$0$&$0$&$1$&$0$\\
\noalign{\bigskip}
\multicolumn{16}{c}%
{\textbf{Table 6}}\\
\end{tabular}}
\vskip 2pc

In general relativity theory it is postulated that the mass-energy distribution determines the curvature, while in this microscopic picture we assume that it is the other way around. That is, the curvature determines the mass distribution and the curvature is given by the geometry (shell sequence). We are now confronted with the question: What determines the shell sequence of our particular universe? To study this question, the next section presents the macroscopic picture. This picture describes the evolution of a universe as a quantum sequential growth process. In such a process, the probabilities of competing evolutions are determined by quantum amplitudes. Moreover, we shall see the emergence of a discrete 4-manifold. Contrary to prevailing wisdom, the situation may be picturesquely described by the following counterintuitive statement. The microscopic picture is painted by number theory while the macroscopic picture is painted by quantum mechanics. In the end, quantum mechanics determines everything.

\section{Macroscopic Picture} % Section 4
The tree $(\pscript ,\shortrightarrow )$ can be thought of as a growth model and an $x\in\pscript _n$ is a possible universe at step (time) $n$. An instantaneous universe $x\in\pscript _n$ grows one element at a time in one of two different ways. To be specific, if $x\in\pscript _n$ has shell sequence $(s_0(x),s_1(x),\ldots ,s_m(x))$, then $x$ will grow to one of its two offspring $x\to x_0$ or $x\to x_1$ where $x_0$ and $x_1$ have shell sequences
\begin{align*}
&(s_0(x),s_1(x),\ldots ,s_m(x)+1)\\
&(s_0(x),s_1(x),\ldots ,s_m(x),1)
\end{align*}
respectively. We call $x_0$ the $0$-\textit{offspring} and $x_1$ the $1$-\textit{offspring} of $x$. In this way we can recursively order the c-causets in $\pscript$ by using the notation $x_{n,j}$, $n=1,2,\ldots$, $j=0,1,2,\ldots ,2^{n-1}-1$. For example, in terms of their labels we have:
\begin{align*}
x_{1,0}&=(1),\ x_{2,0}=(1,2),\ x_{2,1}=(1;2),\\
x_{3,0}&=(1,2,3),\ x_{3,1}=(1,2;3),\ x_{3,2}=(1;2,3),\ x_{3,3}=(1;2;3),\\
x_{4,0}&=(1,2,3,4),\ x_{4,1}=(1,2,3;4),\ x_{4,2}=(1,2;3,4),\ x_{4,3}=(1,2;3;4),\\
x_{4,4}&=(1;2,3,4),\ x_{4,5}=(1;2,3;4),\ x_{4,6}=(1;2;3,4),\ x_{4,7}=(1;2;3;4).
\end{align*}
In terms of their shell sequences we have:
\begin{align*}
x_{1,0}&=(1),\ x_{2,0}=(2),\ x_{2,1}=(1,1)\\
x_{3,0}&=(3),\ x_{3,1}=(2,1),\ x_{3,2}=(1,2),\ x_{3,3}=(1,1,1),\\
x_{4,0}&=(4),\ x_{4,1}=(3,1),\ x_{4,2}=(2,2),\ x_{4,3}=(2,1,1),\ x_{4,4}=(1,3),\\
x_{4,5}&=(1,2,1),\ x_{4,6}=(1,1,2),\ x_{4,7}=(1,1,1,1)
\end{align*}
In general, the c-causet $x_{n,j}$ has the two offspring $x_{n,j}\to x_{n+1,2j}$ and $x_{n,j}\to x_{n+1,2j+1}$, $n=1,2,\ldots$,
$j=0,1,2,\ldots ,2^{n-1}-1$. For example, $x_{3,2}\to x_{4,4}$ and $x_{3,2}\to x_{4,5}$ while $x_{3,3}\to x_{4,6}$ and $x_{3,3}\to x_{4,7}$. Conversely, for $n=2,3,\ldots$, $x_{n,j}$ has the unique producer $x_{n-1},\floors{j/2}$ where $\floors{j/2}$ is the integer part of $j/2$. For example $x_{5,14}$ has the producer $x_{4,7}$ and $x_{5,13}$ has the producer $x_{4,6}$. With the previous notation $\pscript =\brac{x_{n,j}}$ in place, we call $(\pscript ,\shortrightarrow )$ a \textit{sequential growth process} (SGP).

A c-causet $x_{n,j}$ has a unique history. That is, there exists a unique sequence in $\pscript$ satisfying
\begin{equation*}
x_{1,0}\to x_{2,j_2}\to x_{2,j_3}\to\cdots\to x_{n,j}
\end{equation*}
We say that $x,y\in\pscript$ are \textit{twins} if $x\ne y$ and $x,y$ have the same history. Clearly, $x$ and $y$ are twins if and only if they have the same producer. We can characterize twins as pairs of c-causets of the form $\sqbrac{x_{n,2j},x_{n,2j+1}}$. For example, the twins in the fourth shell are $\sqbrac{x_{4,0},x_{4,1}}$, $\sqbrac{x_{4,2},x_{4,3}}$, $\sqbrac{x_{4,4},x_{4,5}}$, and $\sqbrac{x_{4,6},x_{4,7}}$. Since twins are closely related, we shall find it convenient to identify them. We use the notation $\pscripthat _n$ for the set of twin pairs with $n$ vertices each, $n=2,3,\ldots$ and we let $\pscripthat =\cup\pscripthat _n$ be the set of all twin pairs. Moreover, we use the notation
\begin{equation*}
\xhat _{n,j}=\sqbrac{x_{n,2j},x_{n,2j+1}}, n=2,3,\ldots ,\quad j=0,1,\ldots ,2^{n-2}-1
\end{equation*}
and note that $\ab{\pscript _n}=2^{n-2}$.

We condense the tree $(\pscript ,\shortrightarrow )$ to form the multigraph $(\pscripthat ,\shortrightarrow )$. The term multigraph is used because now there are two edges linking vertices instead of one. We can say that for $\xhat _{n,j}$ we have $\xhat _{n,j}\to\xhat _{n+1,2j}$ and
$\xhat _{n,j}\to\xhat _{n+1,2j+1}$, that is
\begin{align*}
&\sqbrac{x_{n,2j},x_{n,2j+1}}\to\sqbrac{x_{n+1,4j},x_{n+1,4j+1}}\\
\intertext{and}
&\sqbrac{x_{n,2j},x_{n,2j+1}}\to\sqbrac{x_{n+1,4j+2},x_{n+1,4j+3}}
\end{align*}
But now $\xhat _{n,j}\to\xhat _{n+1,2j}$ has two edges $(x_{n,2j},x_{n+1,4j})$ and $(x_{n,2j},x_{n+1,4j+1})$ while
$\xhat _{n,j}\!\to\!\xhat _{n+1,2j+1}$ has two edges $(x_{n,2j+1},x_{n+1,4j+2})$ and $(x_{n,2j+1},x_{n+1,4j+3})$. We thus have four edges
originating at $\xhat _{n,j}$ which we denote by $e_{n,j}^k$ $k=1,2,3,4$ in the order given above.

As an example, consider the twins $\xhat _{3,1}=\sqbrac{x_{3,2},x_{3,3}}$ which produce $\xhat _{4,2}$ and $\xhat _{4,3}$. We can write this as
$\sqbrac{x_{3,2},x_{3,3}}\to\sqbrac{x_{4,4},x_{4,5}}$ and $\sqbrac{x_{3,2},x_{3,3}}\to\sqbrac{x_{4,6},x_{4,7}}$. We then have the following four edges originating at $\xhat _{3,1}$:
\begin{align*}
e_{3,1}^1&=(x_{3,2},x_{4,4}),\quad e_{3,1}^2=(x_{3,2},x_{4,5})\\
e_{3,1}^3&=(x_{3,3},x_{4,5}),\quad e_{3,1}^4=(x_{3,3},x_{4,7})
\end{align*}
We consider $(\pscripthat ,\shortrightarrow )$ as a SGP in which $\shortrightarrow$ is composed of two edges instead of one. We again have that each $\xhat _{n,j}\in\pscripthat$ produces two offspring except that each production can occur in two ways corresponding to the double edges,
$e_{n,j}^1,e_{n,j}^2$ or $e_{n,j}^3,e_{n,j}^4$. Similarly, each $\xhat _{n,j}$ is the offspring of a unique producer but again the production can occur in two ways. We see that $(\pscripthat ,\shortrightarrow )$ resembles a discrete $4$-manifold in which there are four independent ``tangent vectors'' at each $\xhat _{n,j}\in\pscripthat$. The reader may ask why this process is canonical. Why can't we form twins of twins and continue doing this until we finally decide to stop? The reason is that unlike c-causets, twin vertices do not possess unique histories. This is because a history now, not only contains predecessor vertices but also single edges, so an $\xhat _{n,j}\in\pscripthat$ has many histories, in general. We shall consider histories in detail in the next paragraph when we discuss paths.

The twin c-causets $x_{n,2j},x_{n,2j+1}$ are so similar that we identify them and consider $\xhat _{n,j}=\sqbrac{x_{n,2j},x_{n,2j+1}}$ as a single unit. We then view $\xhat _{n,j}\in\pscripthat$ as a possible universe at step $n$, $n=2,3,\ldots\,$. Now two edges are \textit{adjacent} if one enters a twin and the other exits the same twin. Each edge $e _{n,j}^k$ is adjacent to four other edges. To be precise, $e_{n,j}^1$ and $e_{n,j}^2$ are adjacent to $e_{n+1,2j}^k$, $k=1,2,3,4$ while $e_{n,j}^3$ and $e_{n,j}^4$ are adjacent to $e_{n+1,2j+1}^k$, $k=1,2,3,4$. A \textit{path} $\omega$ in
$\pscripthat$ is a sequence of pairwise adjacent edges
\begin{equation*}
\omega =e_{2,0}^{k_2}e_{3,j_3}^{k_3}e_{4,j_4}^{k_4}\cdots
\end{equation*}
An $n$-\textit{path} is a finite sequence of pairwise adjacent edges
\begin{equation*}
\omega =e_{2,0}^{k_2}e_{3,j_3}^{k_3}\cdots e_{n,j_n}^{k_n}
\end{equation*}
We say that an $n$-path $\omega =\omega _2\cdots\omega _n$ has \textit{final vertex} $\xhat\in\pscripthat _{n+1}$ if $\omega _n$ ends at $\xhat$. For example, a $4$-path with final vertex $\xhat _{5,3}$ is $e_{2,0}^2e_{3,0}^3e_{4,1}^4$. Notice that there are $4^{n-1}$ possible $n$-paths. We denote the set of paths in $\pscripthat$ by $\Omegahat$ and the set of $n$-paths by $\Omegahat _n$. We interpret a path as a completed universe (or history) of an evolved universe. A path or $n$-path \textit{contains} $\xhat\in\pscripthat$ if there is an edge of $\omega$ that enters or exits
$\xhat$. We then write $\xhat\in\omega$. In particular $\xhat\in\pscripthat _{n+1}$ is the final vertex of $\omega\in\Omegahat _n$ if
$\xhat\in\omega$.

Let $\escripthat$ be the set of all edges in $\pscripthat$. That is
\begin{equation*}
\escripthat =\brac{e_{n,j}^k\colon n=2,3,\ldots ,\ j=0,1,\ldots ,2^{n-2}-1,\ k=1,2,3,4}
\end{equation*}
A \textit{transition amplitude} is a map $\atilde\colon\escript\to\complex$ satisfying $\sum _{k=1}^4\atilde (e_{n,j}^k)=1$ for all $n,j$. Corresponding to $\atilde$ we define the \textit{amplitude} $a(\omega )$ of an $n$-path $\omega =\omega _2\omega _3\cdots\omega _n\in\Omega _n$ to be a
$a(\omega )=\atilde (\omega _2)\atilde (\omega _3)\cdots\atilde (\omega _n)$. The \textit{amplitude} of a set $A\subseteq\Omegahat _n$ is
\begin{equation*}
a(A)=\sum\brac{a(\omega )\colon\omega\in A}
\end{equation*}
Notice that $a(\Omegahat _n)=1$. The \textit{amplitude} $a(\xhat _{n,j})$ of $\xhat _{n,j}\in\pscripthat _n$ is defined by
\begin{equation*}
a(\xhat _{n,j})=\sum\brac{a(\omega )\colon\omega\in\Omegahat _{n-1},\xhat _{n,j}\in\omega}
\end{equation*}
It follows that $\sum\brac{a(\xhat )\colon\xhat\in\pscripthat _n}=1$. We call $c_{n,j}^k=\atilde (e_{n,j}^k)$ \textit{coupling constants} and note that
$\sum _{k=1}^4c_{n,j}^k=1$ for all $n,j$.

Letting $\ascript _n=2^{\Omegahat _n}$ be the power set on $\Omegahat _n$, we have that $(\Omegahat _n,\ascript _n)$ is a measurable space. The $n$-\textit{decoherence functional} $D_n\colon\ascript _n\times\ascript _n\to\complex$ corresponding to $\atilde$ is defined by
$D_n(A,B)=\overline{a(A)}a(B)$. We interpret $D_n(A,B)$ as a measure of the quantum interference between the events $A$ and $B$. The
$n$-\textit{decoherence matrix} is the $4^{n-1}\times 4^{n-1}$ matrix with components $M_n(\omega ,\omega ')=\overline{a(\omega )}a(\omega ')$,
$\omega ,\omega '\in\Omegahat _n$. Letting $\ket{a(\omega )}$ $\omega\in\Omegahat _n$, be the $4^{n-1}$ dimensional vector with components $a(\omega )$, we see that $M_n$ is the rank~1 positive operator given by $M_n=\ket{a(\omega )}\bra{a(\omega )}$. Again, $D_n(\omega ,\omega ')$ is interpreted as a measure of the interference between the paths $\omega$ and $\omega '$. The $q$-\textit{measure} of
$A\subseteq\Omegahat _n$ corresponding to $\atilde$ is defined by $\mu _n(A)=D_n(A,A)=\ab{a(A)}^2$. In particular,
$\mu _n(\omega )=\ab{a(\omega )}^2$, $\omega\in\Omegahat _n$ and $\mu _n(\xhat _{n,j})=\ab{a(\xhat _{n,j})}^2$.

We interpret $\mu _n(A)$ as the quantum propensity of the event $A\in\ascript _n$. The $q$-measure $\mu _n$ is determined by the coupling constants $c_{n,j}^k$. It is believed that once the coupling constants, and hence the $q$-measures $\mu _n$ are known, then certain $n$-paths and c-causet twins $\xhat _{n,j}$ will have dominate propensities. In this way we will determine dominate geometries for the microscopic picture of our particular universe. The $q$-measure $\mu _n$ is not a measure on the $\sigma$-algebra $\ascript _n$, in general \cite{sor94, sor03, sur11}. This is because the additivity condition $\mu _n(A\cup B)=\mu _n(A)+\mu _n(B)$ whenever $A\cap B=\emptyset$, is not satisfied, in general. The physical reason for this additivity failure is quantum interference. However, $\mu _n$ does satisfy a more general condition called
\textit{grade}-2 \textit{additivity} defined as follows. If $A,B,C\in\ascript _n$ are mutually disjoint, then
\begin{equation*}
\mu _n(A\cup B\cup C)=\mu _n(A\cup B)+\mu _n(A\cup C)+\mu _n(B\cup C)-\mu _n(A)-\mu _n(B)-\mu _n(C)
\end{equation*}
Of course, we do have that $\mu _n(\Omegahat _n)=1$ and $\mu _n(A)\ge 0$ for all $A\in\ascript _n$
\bigskip

\begin{exam}{1}
In this example we compute the $q$-measures of the first few twins in terms of the coupling constants. We introduce the notation
\begin{align*}
&\atilde (\xhat _{n,j}\to\xhat _{n+1,2j})=c_{n,j}^1+c_{n,j}^2\\
&\atilde (\xhat _{n,j}\to\xhat _{n+1,2j+1})=c_{n,j}^3+c_{n,j}^4\\
\end{align*}
We then have
\begin{align*}
\mu _3(\xhat _{3,0})&=\ab{\atilde (\xhat _{2,0}\to\xhat _{3,0})}^2=\ab{c_{2,0}^1+c_{2,0}^2}^2\\
\mu _3(\xhat _{3,1})&=\ab{\atilde (\xhat _{3,0}\to\xhat _{3,1})}^2=\ab{c_{2,0}^3+c_{2,0}^4}^2\\
\mu _4(\xhat _{4,0})&=\ab{\atilde (\xhat _{2,0}\to\xhat _{3,0})}^2\ab{\atilde (\xhat _{3,0}\to\xhat _{4,0})}^2
  =\ab{c_{2,0}^1+c_{2,0}^2}^2\ab{c_{3,0}^1+c_{3,0}^2}^2\\
\mu _4(\xhat _{4,1})&=\ab{\atilde (\xhat _{2,0}\to\xhat _{3,0})}^2\ab{\atilde (\xhat _{3,0}\to\xhat _{4,1})}^2
  =\ab{c_{2,0}^1+c_{2,0}^2}^2\ab{c_{3,0}^3+c_{3,0}^4}^2\\
\mu _4(\xhat _{4,2})&=\ab{\atilde (\xhat _{2,0}\to\xhat _{3,1})}^2\ab{\atilde (\xhat _{3,1}\to\xhat _{4,2})}^2
  =\ab{c_{2,0}^3+c_{2,0}^4}^2\ab{c_{3,1}^1+c_{3,1}^2}^2\\
\mu _4(\xhat _{4,3})&=\ab{\atilde (\xhat _{2,0}\to\xhat _{3,0})}^2\ab{\atilde (\xhat _{3,1}\to\xhat _{4,3})}^2
  =\ab{c_{2,0}^3+c_{2,0}^4}^2\ab{c_{3,1}^3+c_{3,1}^4}^2\\
\end{align*}
\end{exam}

In general, $\mu _n\paren{\brac{\xhat _{n,i},\xhat _{n,j}}}\ne\mu _n(\xhat _{n,i})+\mu _n(\xhat _{n,j})$. For instance,\newline in Example 1
\begin{equation*}
1=\mu _3\paren{\brac{\xhat _{3,0},\xhat _{3,1}}}\ne\ab{c_{2,0}^1+c_{2,0}^2}^2+\ab{c_{2,0}^3+c_{2,0}^4}^2
  =\mu _3(\xhat _{3,0})+\mu _3(\xhat _{3,1})
\end{equation*}
in general. When we have $\mu _n\paren{\brac{\xhat _{n,i},\xhat _{n,j}}}=\mu _n(\xhat _{n,i})+\mu _n(\xhat _{n,j})$, we say that $\xhat _{n,i}$ and
$\xhat _{n,j}$ \textit{do not interfere}. When this equality is replaced by $<$ ($>$) we say that $\xhat _{n,i}$ and $\xhat _{n,j}$
\textit{interfere destructively} (\textit{constructively}). It is useful to notice that if $\xhat _{n,i}\to\xhat _{n+1,j}$ then
\begin{equation*}
\mu _{n+1}(\xhat _{n+1,j})=\mu _n(\xhat _{n,i})\ab{\atilde (x_{n,k}\to x_{n+1,j})}^2
\end{equation*}

\section{An Amplitude Process} % Section 5
We call an amplitude $\atilde\colon\escript\to\complex$ an \textit{amplitude process} (AP). Corresponding to an AP $\atilde$ we have the coupling constants $c_{n,j}^k=\atilde (e_{n,j}^k)$ and these determine the $q$-measures $\mu _n$. This section presents a simple but nontrivial example of an AP that may have some physical significance.
We define coupling constants $c_{n,j}^k$, $n=2,3,\ldots$, $j=0,1,\ldots ,2^{n-2}-1$, $k=1,2,3,4$ as follows. Let $\pi /16\le\theta _{n,j}\le\pi /12$ be the angle given by
\begin{equation*}
\theta _{n,j}=\sqbrac{1-\paren{\tfrac{1}{2}-\frac{j}{2^{n-2}}}^2}\frac{\pi}{12}
\end{equation*}
and define $c_{n,j}^1=e^{i\theta n,j}/z_{n,j}$, $c_{n,j}^2=ic_{n,j}^1$, $c_{n,j}^3=\overline{c_{n,j}^1}$, $c_{n,j}^4=\overline{c_{n,j}^2}$ where
\begin{align*}
z_{n,j}&=e^{i\theta n,j}(1+i)+e^{-i\theta n,j}(1-i)=2(\cos\theta _{n,j}-\sin\theta _{n,j})\\
&=2\sqrt{2}\,\sin\paren{\frac{\pi}{4}-\theta _{n,j}}
\end{align*}
It is useful to note that $z_{n,j}^2=4(1-\sin 2\theta _{n,j})$. As we have seen in Example 1, it is also useful to know
\begin{align*}
\ab{c_{n,j}^1+c_{n,j}^2}^2&=\ab{e^{i\theta _{n,j}}+ie^{i\theta _{n,j}}}^2/\ab{z_{n,j}}^2=\frac{\ab{1+i}^2}{4(1-\sin 2\theta _{n,j})}\\
&=\frac{1}{2(1-\sin 2\theta _{n,j})}
\end{align*}
and $\ab{c_{n,j}^3+c_{n,j}^4}^2=\ab{c_{n,j}^1+c_{n,j}^2}^2$.
\bigskip

\begin{exam}{2}
We compute the $q$-measures of the first few twins. For $n=2,3,4$ we have:
$\theta _{2,0}=\pi /16$, $\theta _{3,0}=\pi /16$, $\theta _{3,1}=\pi/12$, $\theta _{4,0}=\pi /16$, $\theta _{4,1}=(15/16)(\pi /12)$,
$\theta _{4,3}=(15/16)(\pi /12)$. It follows from Example 1, that
\begin{equation*}
\mu _3(\xhat _{3,0})=\mu _3(\xhat _{3,1})=\frac{1}{2(1-\sin\pi /8)}=0.80996
\end{equation*}
Moreover,
\begin{align*}
\mu _4(\xhat _{4,0})&=\mu _4(\xhat _{4,1})=\frac{1}{4(1-\sin\pi /8)^2}=0.65603\\
\intertext{and}
\mu _4(\xhat _{4,2})&=\mu _4(\xhat _{4,3})=\frac{1}{2(1-\sin\pi /8)}\cdot\frac{1}{2(1-\sin\pi /6)}=0.80996
\end{align*}
\end{exam}
\bigskip

Notice that $z_{n,j}$ attains a minimum when $\theta _{n,j}=\pi /12$. This minimum occurs when $j=2^{n-3}$, $n=3,4,\ldots\,$. Using the notation
$j_n=2^{n-3}$, $n=3,4,\ldots$, we have $j_3=1, j_4=2, j_6=8,\ldots\,$. It follows that the twins of maximum propensity are: 
$\xhat _{3,0},\xhat _{3,1}; \xhat _{4,2},\xhat _{4,3};\xhat _{5,4},\xhat _{5,5};\xhat _{6,8},\xhat _{6,9};\ldots\,$. That is,
$\xhat _{n,2^{n-3}},\xhat _{n,2^{n-3}+1}$ have the maximum propensity $0.80996$ of being an actual realized universe. We conclude that in this model the twins with highest propensity lie in the ``middle'' of the process $(\pscripthat ,\shortrightarrow )$. Twins close to the ``middle'' also have high propensity and the $q$-measure decreases to close to zero for large $n$ as we move toward the right or left of the multigraph. This is more easily seen examining the $n=5$ case in the next example.
\bigskip

\begin{exam}{3}
The $q$-measures for $n=5$ are now computed.
\begin{align*}
\mu _5(\xhat _{5,0})&=\mu (\xhat _{5,1})=\mu _4(\xhat _{4,0})\ab{\atilde (\xhat _{4,0}\to\xhat _{5,0})}^2=(0.65603)(0.80996)\\
  &=0.53136\\
\mu _5(\xhat _{5,2})&=\mu _5(\xhat _{5,3})=\mu _4(\xhat _{4,1})\ab{\atilde (\xhat _{4,1}\to\xhat _{5,2})}^2=(0.65603)(0.94589)\\
  &=0.62035\\
  \mu _5(\xhat _{5,4})&=\mu _5(\xhat _{5,5})=\mu _4(\xhat _{4,2})\ab{\atilde (\xhat _{4,2}\to\xhat _{5,4})}^2=0.80996\\
\mu _5(\xhat _{5,6})&=\mu (\xhat _{5,7})=\mu _4(\xhat _{4,3})\ab{\atilde (\xhat _{4,3}\to\xhat _{5,6})}^2=(0.80996)(0.94589)\\
  &=0.76613\\
\end{align*}
\end{exam}

We now consider $q$-measures of paths. The $2$-paths are $e_{2,0}^k$, $k=1,2,3,4$, and we have
\begin{equation*}
\mu _2(e_{2,0}^k)=\frac{1}{(z _{2,0})^2}=\frac{1}{4(1-\sin\pi /8)}=0.40498
\end{equation*}
Since
\begin{align*}
\mu _2\paren{\brac{e_{2,0}^1,e_{2,0}^2}}&=\ab{a\paren{\brac{e_{2,0}^1,e_{2,0}^2}}}^2=\ab{\atilde (e_{2,0}^1)+\atilde (e_{2,0}^2)}^2\\
\noalign{\medskip}
   &=\frac{\ab{1+i}^2}{(z_{2,0})^2}=\frac{2}{(z_{2,0})^2}=\mu _2(e_{2,0}^1)+\mu _2(e_{2,0}^2)
\end{align*}
we conclude that $e_{2,0}^1$ and $e_{2,0}^2$ do not interfere. Similarly, $e_{2,0}^3$ and $e_{2,0}^4$ do not interfere. Since
\begin{align*}
\mu _2\paren{\brac{e_{2,0}^1,e_{2,0}^3}}&=\ab{\atilde (e_{2,0}^1)+\atilde (e_{2,0}^3)}^2=\frac{\ab{e^{i\pi /16}+e^{-i\pi /16}}^2}{(z_{2,0})^2}\\
   &=\frac{\cos ^2\pi /16}{1-\sin\pi /8}=1.5583
\end{align*}
and $\mu _2(e_{2,0}^1)+\mu _2(e_{2,0}^3)=0.80996$ we see that $e_{2,0}^1$ and $e_{2,0}^3$ interfere constructively. Moreover,
\begin{align*}
\mu _2\paren{\brac{e_{2,0}^2,e_{2,0}^4}}&=\ab{\atilde (e_{2,0}^2)+\atilde (e_{2,0}^4)}^2=\frac{\ab{e^{i\pi /16}-e^{-i\pi /16}}^2}{(z_{2,0})^2}\\
   &=\frac{\sin ^2\pi /16}{1-\sin\pi /8}=0.061654
\end{align*}
so $e_{3,0}^2$ and $e_{2,0}^4$ interfere destructively. Also,
\begin{align*}
\mu _2\paren{\brac{e_{2,0}^1,e_{2,0}^4}}&=\ab{\atilde (e_{2,0}^1)+\atilde (e_{2,0}^4)}^2=\frac{\ab{e^{i\pi /16}-ie^{-i\pi /16}}^2}{(z_{2,0})^2}\\
\noalign{\medskip}
   &=\frac{2(\cos\pi /16-\sin\pi /16)^2}{(z_{2,0})^2}=\frac{2(1-\sin\pi /8)}{(z_{2,0})^2}=0.5000
\end{align*}
Hence, $e_{2,0}^1$ and $e_{2,0}^4$ interfere destructively. In a similar way $e_{2,0}^2$ and $e_{2,0}^3$ interfere destructively because
$\mu _2\paren{\brac{e_{2,0}^2,e_{2,0}^3}}=0.5000$.

We now consider $q$-measures of longer paths. Let $\omega _n\in\Omegahat _n$ be given by
\begin{equation*}
\omega _n=e_{2,0}^{k_2}e_{3,j_3}^{k_3}\cdots e_{n,j_n}^{k_n},\quad j_i\in\brac{0,1,\ldots ,2^{i-1}-1},\ k_i\in\brac{1,2,3,4}
\end{equation*}
The $q$-measure of $\omega _n$ becomes
\begin{align*}
\mu _n(\omega _n)&=\ab{c_{2,0}^{k_2}}^2\ab{c_{3,j_3}^{k_3}}^2\cdots\ab{c_{n,j_n}^{k_n}}^2
   =\frac{1}{\ab{z_{2,0}}^2}\,\frac{1}{\ab{z_{3,j_3}}^2}\cdots\frac{1}{\ab{z_{n,j_n}}^2}\\
\noalign{\medskip}
   &=\frac{1}{4^{n-1}}\frac{1}{(1-\sin 2\theta _{2,0})}\frac{1}{(1-\sin 2\theta _{3,j_3})}\cdots\frac{1}{(1-\sin 2\theta _{n,j_n})}
\end{align*}
As $n\to\infty$ we see that $\mu _n(\omega _n)\to 0$ for any $n$-path $\omega _n$. However, as with twins, the paths toward the ``middle'' of
$\Omegahat _n$ have higher propensities than the others. These paths with the highest propensity have $\theta _{i,j_i}=\pi /12$,
$i=3,4,\ldots$ in which case $j_n=2^{n-3}$. Hence, they have the form
\begin{equation*}
\omega _n=e_{2,0}^{k_2}e_{3,1}^{k_3}e_{4,2}^{k_4}\cdots e_{n,2^{n-3}}^{k_n},\quad k_2\in\brac{3,4},k_3,k_4,\ldots ,k_n\in\brac{1,2}
\end{equation*}
We conclude that
\begin{equation*}
\mu _n(\omega _n)=\frac{1}{4^{n-1}}(0.40498)2^{n-2}=\frac{0.40498}{2^n}
\end{equation*}
Since there are $4^{n-1}$ $n$-paths in all, this is a dominate propensity. Notice there are $2^{n-1}$ such high propensity $n$-paths.

Let $A_n\subseteq\Omegahat _n$ be the set of $n$-paths with highest propensity and let $A\subseteq\Omegahat$ be the extensions of such
$n$-paths to infinity. We consider $A$ as the set of most likely universes. Now
\begin{align*}
\mu _n(A_n)&=\ab{c_{2,0}^3+c_{2,0}^4}^2\ab{c_{3,1}^1+c_{3,1}^2}^2\cdots\ab{c_{n,2^{n-3}}^1+c_{n,2^{n-3}}^2}^2\\
   &=\frac{1}{2(1-\sin\pi /8)}=0.80996
\end{align*}
In the language of \cite{gud13}, we say $A$ is a \textit{suitable} subset of $\Omegahat$ with $q$-measure $\mu (A)=0.80996$. This result states that in the present model, our particular universe is likely to be in the set $A$ of possible universes, or it may deviate slightly from $A$.

It is also of interest to consider the set $A'_n$ of paths of lower propensity. To compute $\mu _n(A'_n)$ we have that
\begin{align*}
a(A_n)&=(c_{2,0}^3+c_{2,0}^4)(c_{3,1}^1+c_{3,1}^2)\cdots (c_{n,2^{n-3}}^1+c_{n,2^{n-3}}^2)\\
\noalign{\medskip}
   &=\frac{e^{-i\pi /16}(1-i)\sqbrac{e^{i\pi /12}(1+i)}^{n-2}}{(2\sqrt{2}\sin 3\pi /16)(\sqrt{2})^{n-2}}=0.89998e^{i47\pi /48}e^{in\pi /3}
\end{align*}
Hence,
\begin{align*}
\mu (A'_n)&=\ab{a(A'_n)}^2=\ab{1-a(A_n)}^2=1+\ab{a(A_n)}^2-2\itre a(A_n)\\
  &=1.80996-1.79996\cos\paren{\frac{47}{48}+\frac{n}{3}}\pi
\end{align*}
The quantity $\mu (A'_n)$ oscillates and does not have a limit as $n\to\infty$. However, using the methods of \cite{gud13}, one can show that $A'\subseteq\Omegahat$ is suitable with $q$-measure $\mu (A')=1$. This is not surprising because there are predominately more elements in
$A'_n$ than in $A_n$ for large $n$.

\section{Covariant Difference Operators} % Section 6
We have seen that in the macroscopic picture the SGP $(\pscripthat ,\shortrightarrow )$ resembles a discrete $4$-manifold. In this section we briefly examine this resemblance by introducing difference operators. As before we denote the coupling constants by $c_{n,j}^k=\atilde (e_{n,j}^k)$, $n=2,3,\ldots$, $j=0,1,\ldots ,2^{n-2}-1$, $k=1,2,3,4$. To avoid trivialities, we assume that $c_{n,j}^1\ne c_{n,j}^2$ and $c_{n,j}^3\ne c_{n,j}^4$. Let $H=L_2(\pscripthat )$ be the Hilbert space of square summable complex-valued functions on $\pscripthat$ with the usual inner product
\begin{equation*}
\elbows{f,g}=\sum _{\xhat\in\pscript}\overline{f(\xhat )}g(\xhat )
\end{equation*}
We define the \textit{covariant difference operators} $\nabla _k$, $k=1,2,3,4$ on $H$ by
\begin{equation*}
\nabla _kf(\xhat _{n,j})=f(\xhat _{n+1,2j+\floors{k/3}})-c_{n+1,2j+\floors{k/3}}^kf(\xhat _{n,j})
\end{equation*}
Globally, the operators $\nabla _k$, $k=1,2,3,4$ are linearly independent. However, locally only three of them are linearly independent as the following result shows.
\begin{thm}       % Theorem 6.1
\label{thm61}
{\rm (i)}\enspace The operators $\nabla _1,\nabla _2,\nabla _3$ are linearly independent at each $\xhat _{n,j}\in\pscripthat$.
{\rm (ii)}\enspace At $\xhat _{n,j}\in\pscripthat$ we have
\begin{equation*}
\nabla _4=\frac{c_{n+1,2j+1}^4-c_{n+1,2j+1}^3}{c_{n+1,2j}^1-c_{n+1,2j}^2}(\nabla _1-\nabla _2)+\nabla _3
\end{equation*}
\end{thm}
\begin{proof}
(i)\enspace Suppose $a_1\nabla _1+a_2\nabla _2+a_3\nabla _3=0$ at $\xhat _{n,j}$. We then have for any $f\in H$ that
\begin{align*}
(a_1+a_2)f&(\xhat _{n+1,2j})+a_3f(\xhat _{n+1,2j+1})\\
   &-(a_1c_{n+1,2j}^1+a_2c_{n+1,2j}^2+a_3c_{n+1,2j+1}^3)f(\xhat _{n,j})=0
\end{align*}
Letting $f\in H$ satisfy $f(\xhat _{n+1,2j})=1$, $f(\xhat _{n+1,2j+1})=f(\xhat _{n,j})=0$ we conclude that $a_1+a_2=0$. In a similar way we have that $a_3=0$ and $a_1c_{n+1,2j}^1-a_1c_{n+1,2j}^2=0$. Since $c_{n+1,2j}^1\ne c_{n+1,2j}$, we conclude that $a_1=0$ so $a_2=0$.\newline
(ii)\enspace For any $f\in H$ we have 
\begin{align*}
&\frac{c_{n+1,2j+1}^4\!-\!c_{n+1,2j+1}^3}{c_{n+1,2j}^1\!-\!c_{n+1,2j}^2}\ (\nabla _1\!-\!\nabla _2)f(\xhat _{n,j})+\nabla _3f(\xhat _{n,j})\\
&=\frac{c_{n+1,2j+1}^4\!-\!c_{n+1,2j+1}^3}{c_{n+1,2j}^1\!-\!c_{n+1,2j}^2}(c_{n+1,2j}^2\!-\!c_{n+1,2j}^1)f(\xhat _{n,j})
\!+\!f(\xhat _{n+1,2j+1})\!-\!c_{n+1,2j+1}^3f(\xhat _{n,j})\\
&=f(\xhat _{n+1,2j+1})-c_{n+1,2j+1}^4f(\xhat _{n,j})=\nabla _4f(\xhat _{n,j})\qedhere
\end{align*}
\end{proof}

In a similar way, any three of the operators $\nabla _k$ are locally linearly independent and any one of them can locally be written as a linear combination of the other three.

Let $\omega =e_{2,0}^{k_2}\cdots e_{n,j}^{k_n}\cdots$ be a path. We write $k(\omega ,n)=k_n$ and call $k_n$ the \textit{direction of} $\omega$
\textit{at} $n$. We also define
\begin{equation*}
a_\omega (\xhat _{n,j})=c_{2,0}^{k_2}\cdots c_{n-1}^{k_{n-1}}
\end{equation*}
if $\xhat _{n,j}\in\omega$ and $a_\omega (\xhat _{n,j})=0$, otherwise. The $\omega$-\textit{covariant difference operator} is the operator
$\nabla _\omega$ on $H$ given by
\begin{equation*}
\nabla _\omega f(\xhat _{n,j})=a_\omega (\xhat _{n,j})\nabla _{k(\omega ,n)}f(\xhat _{n,j})
\end{equation*}
For all $f\in H$ we have
\begin{equation*}
\nabla _\omega f(\xhat _{n,j})
  =a_\omega (\xhat _{n,j})f(\xhat _{n+1,2j+\floors{k(\omega ,n)/3}})-a_\omega (\xhat _{n+1,2j+\floors{k(\omega ,n)/3}})f(\xhat _{n,j})
\end{equation*}
Now $a_\omega\in H$ and it follows that $\nabla _\omega a_\omega =0$ which is why $\nabla _k$ and $\nabla _\omega$ are called covariant.

We now extend this formalism to functions of two variables. Let $K=H\otimes H$ which we can identify with $L_2(\pscripthat\times\pscripthat )$. For $k,k'=1,2,3,4$, we define the \textit{covariant bidifference operators} $\nabla _{k,k'}\colon K\to K$ by
\begin{align*}
\nabla _{k,k'}f(\xhat _{n,j},\xhat _{n',j'})&=f(\xhat _{n+1,2j+\floors{k/3}},\xhat _{n'+1,2j'+\floors{k'/3}})\\
  &\quad -\cbar _{n+1,2j+\floors{k/3}}^kc_{n'+1,2j'+\floors{k'/3}}^{k'}f(\xhat _{n,j},\xhat _{n'j'})
\end{align*}
For $\omega ,\omega '\in\Omegahat$, the $\omega\omega '$-\textit{covariant bidifference operator} is the operator\newline
$\nabla _{\omega\omega '}\colon K\to K$ given by
\begin{align*}
\nabla _{\omega\omega '}f&(\xhat _{n,j},\xhat _{x',j'})
  =\abar _\omega (\xhat _{n,j})a_{\omega '}(\xhat _{n',j'})\nabla _{k(\omega ,n),k(\omega ',n')}f(\xhat _{n,j},\xhat _{n',j'})\\
  &=\abar _\omega (\xhat _{n,j})a_{\omega '}(\xhat _{n',j'})f(\xhat _{n+1,2j+\floors{k(\omega ,n)/3}},
  \xhat _{n'+1,2j'+\floors{k(\omega ',n')/3}})\\
  &\quad -\abar _\omega (\xhat _{n+1,2j+\floors{k(\omega ,n)/3}})a_{\omega '}(\xhat _{n'+1,2j'+\floors{k(\omega ',n')/3}})f(\xhat _{n,j},\xhat _{n',j'})
\end{align*}
Again, $\nabla _{\omega ,\omega '}$ is called covariant because we have $\nabla _{\omega ,\omega '}\abar _\omega a_{\omega '}=0$.

The \textit{curvature operator} is defined by $\rscript _{\omega ,\omega '}=\nabla _{\omega ,\omega '}-\nabla _{\omega ',\omega}$. We thus have
\begin{align*}
\rscript _{\omega ,\omega '}f(\xhat _{n,j},\xhat _{n',j'})&=\abar _\omega (\xhat _{n,j})a_{\omega '}(\xhat _{n',j'})
  \nabla _{k(\omega ,n),k(\omega ',n')}f(\xhat _{n,j},\xhat _{n',j'})\\
  &\quad -\abar _{\omega '}(\xhat _{n,j})a_\omega (\xhat _{n',j'})\nabla _{k(\omega ',n'),k(\omega ,n)}f(\xhat _{n,j},\xhat _{n',j'})
\end{align*}
Define the operators $\dscript _{k,k'}\colon K\to K$ by
\begin{align*}
\dscript _{k,k'}f(\xhat _{n,j},\xhat _{n',j'})
  &=\abar _\omega (\xhat _{n,j})a_{\omega '}(\xhat _{n',j'})f(\xhat _{n+1,2j+\floors{k/3}},\xhat _{n'+1,2j'+\floors{k'/3}})\\
  &\quad -\abar _{\omega '}(\xhat _{n,j})a_\omega (\xhat _{n',j'})f(\xhat _{n'+1,2j'+\floors{k'/3}},\xhat _{n+1,2j+\floors{k/3}})
\end{align*}
and the operators $\tscript _{k,k'}\colon K\to K$ by
\begin{align*}
\tscript _{k,k'}f(\xhat _{n,j},\xhat _{n',j'})
  &=\left[\abar _{\omega '}(\xhat _{n+1,2j+\floors{k/3}})a_\omega (\xhat _{n'+1,2j'+\floors{k'/3}})\right.\\
  &\quad \left.-\abar _\omega (\xhat _{n+1,2j+\floors{k/3}})a_{\omega '}(\xhat _{n'+1,2j'+\floors{k'/3}})\right]f(\xhat _{n,j},\xhat _{n',j'})
  \end{align*}
If we define $\dscript _{\omega ,\omega '}$ and $\tscript _{\omega ,\omega '}$ by
\begin{align*}
\dscript _{\omega ,\omega '}f(\xhat _{n,j},\xhat _{n',j'})=\dscript _{k(\omega ,n),k(\omega ',n')}f(\xhat _{n,j},\xhat _{n',j'})\\
\intertext{and}
\tscript _{\omega ,\omega '}f(\xhat _{n,j},\xhat _{n',j'})=\dscript _{k(\omega ,n),k(\omega ',n')}f(\xhat _{n,j},\xhat _{n',j'})\\
\end{align*}
then it is easy to check that
\begin{equation}         % equation (6.1)
\label{eq61}
\rscript _{\omega ,\omega '}=\dscript _{\omega ,\omega '}+\tscript _{\omega ,\omega '}
\end{equation}
We call \eqref{eq61} the \textit{discrete Einstein Equations}. For a further discussion of these equations, we refer the reader to \cite{gud12}.


\begin{thebibliography}{99}
% ref 1
\bibitem{gtw09}Y.~Ghazi-Tabatabai and P.~Wallden, 
Dynamics and predictions in the co-event interpretation, \textit{J. Phys.~A} \textbf{42} 235303 (2009).
% ref 2
\bibitem{gud12}S.~Gudder, An Einstein equation for discrete quantum gravity, arXiv: gr-gc 1204.4506 (2012).
% ref 3
\bibitem{gud13}S.~Gudder, An approach to discrete quantum gravity, arXiv: gr-gc 1305.5184v1 (2013).
% ref 4
\bibitem{gudp13}S.~Gudder, The universe as a quantum computer, in preparation.
% ref 5
\bibitem{hen09}J.~Henson, Quantum histories and quantum gravity, arXiv: gr-qc 0901.4009 (2009).
% ref 6
\bibitem{sor94}R.~Sorkin, 
Quantum mechanics as quantum measure theory, \textit{Mod. Phys. Letts.~A} \textbf{9} (1994), 3119--3127.
% ref 7
\bibitem{sor03}R.~Sorkin, Causal sets: discrete gravity, arXiv: gr-qc 0309009 (2003).
% ref 8
\bibitem{sur11}S.~Surya, Directions in causal set quantum gravity, arXiv: gr-qc 1103.6272 (2011).
% ref 9
\bibitem{wal13}P.~Wallden,
The coevent formulation of quantum theory, arXiv: gr-qc 1301.5704 (2013).
\end{thebibliography}
\end{document}